\documentclass[runningheads]{llncs}
\usepackage{graphicx}
\usepackage{comment}
\usepackage{algorithm}
\usepackage[noend]{algpseudocode}

\newtheorem{pro}{Problem}

\def\for{\mbox{\it for}}
\def\rf{\mbox{\it rf}}
\def\Vp{\mbox{\it Vp}}
\def\VP{\mbox{\it VP}}
\def\VL{\mbox{\it VL}}

\def\TI{\mbox{\it TI}}
\def\WVP{\mbox{\it WVP}}

\def\OPT{\mbox{\it OPT}}
\def\SCR{\mbox{\it SCR}}
\def\EdP{\mbox{\it EdP}}
\def\EdI{\mbox{\it EdI}}
\def\inPL{\mbox{\it inPL}}
\def\seg#1{\overline{#1}}
\def\P{\cal P}

\def\S{\cal S}
\def\F{\cal F}
\begin{document}
\title{Art Gallery Plus Single Specular-Reflection}

\author{Arash Vaezi\inst{1}\email{avaezi@ce.sharif.edu}\\ \and
Bodhayan Roy\inst{2}\email{broy@maths.iitkgp.ac.in}\Comment{supported by an ISIRD Grant from Sponsored Research and Industrial Consultancy, IIT Kharagpur}\and
\\ Mohammad Ghodsi\inst{1,3}\email{ghodsi@sharif.edu}\Comment{This author's research was partially supported IPM under grant No: CS1392-2-01.}
}
\authorrunning{A. Vaezi et al.}

\institute{Sharif University of Technology, Azadi Ave
Tehran, Iran \and
Indian Institute of Technology, Kharagpur \and
Institute for Research in Fundamental Sciences (IPM)
}
\maketitle              % typeset the header of the contribution
\begin{abstract}
Given a simple polygon $\P$, in the Art Gallery problem, the goal is to find the minimum number of guards needed to cover the entire $\P$, where a guard is a point and can see another point $q$ when $\overline{pq}$ does not cross the edges of $\P$.
 This paper studies a variant of the Art Gallery problem in which the boundaries of $\P$ are replaced by single specular-reflection edges, allowing the view rays to reflect once per collision with an edge. This property allows the guards to see through the reflections, thereby viewing a larger portion of the polygon.

For this problem, the position of the guards in $\P$ can be determined with our proposed $\mathcal{O}(\log n)$-approximation algorithm. Besides presenting an algorithm with the mentioned approximation factor, we will see that reflection can decrease the number of guards in practice. The proposed algorithm converts the generalized problem to the Set Cover problem.

\end{abstract}

\section{Introduction and Related Works}
Consider a simple polygon $\P$ with $n$ vertices.   
The maximal sub-polygon of $\P$ visible to a point $q$ in $\P$ is called the \emph{visibility polygon} of $q$, which is denoted by $\VP(q)$. 
There are linear-time algorithms to compute $\VP(q)$ when the viewer is a point \cite{1}.
 For the segment $\seg{pq}$ inside $\P$, the \emph{weak visibility polygon} of $\seg{pq}$, denoted as $\WVP(\seg{pq})$, is the maximal sub-polygon of $\P$ visible to at least one point (not the endpoints) of $\seg{pq}$.
A polygon $\cal Q$ inside $\P$ is said to be \emph{completely visible} from $\seg{pq}$ if for every point $z \in \cal Q$ and for any point $w \in \seg{pq}$, $w$ and $z$ are visible (denoted as $\mathit{CVP}$ short from the completely visible polygon). Also, $\cal Q$ is said to be \emph{strongly visible} from $\seg{pq}$ if there exists a point $w \in \seg{pq}$, such that for every point $z \in \cal Q$, $w$ and $z$ are visible ($\mathit{SVP}$). These different visibility types can be computed in linear time (\cite{2,avis}).

The visibility of a point can only be blocked by a part of the polygon. In a polygon, a vertex is called \emph{convex} if the internal angle of the polygon (i.e., the angle formed by the two edges at the vertex with the polygon inside the angle) is less than $\pi$ radians (180°); otherwise, it is called \emph{reflex}. Reflex-vertices determine the start of the blocked visibility of a point inside $\P$.

The visibility of the point $q$ can be extended if some edges of $\P$ reflect the visibility rays incident on them. 
While different types of reflections have been studied by researchers \cite{5}, we only consider the \emph{specular} type. In the specular reflection, a view ray is reflected into a single outgoing direction, where the angle of the reflection and the incident angle are equal. 
Another well-known type of reflection is called diffuse reflection. A ray reflected via diffuse reflection may assume all angles between $0$ and $\pi$.
This paper only deals with a single reflection per view ray when the reflection is specular.

First, we assume that all edges visible to $q$, the viewer, are reflection edges. Then, we obtain the edges whose visibility can make additional visibility to the viewers. We will find the exact part of every edge whose visibility is beneficial for at least one viewer. To do this, we use the algorithm presented in \cite{tcs}. 
The two points $x$ and $y$ inside $\P$ can see each other through $e$, if and only if they are visible via reflection rays on $e$. We call these points \emph{reflected-visible}.
 ~Sometimes, to be more specific, we use $e$-reflected-visible.
 ~In order to compute the $e$-reflected-visibility area, only \emph{the visible part} of $e$ should be taken into account.

Given a simple polygon $\P$, in the Art Gallery problem, the goal is to determine the minimum number of stationary points, called \emph{guards}, that can be sufficient to see every point in the interior of a given polygon $\P$. This article intends to empower the guards with the single specular reflection edges.
If guards are placed at the vertices of $\P$, they are called \emph{vertex guards}.
If guards are placed at any point of $\P$, they are called \emph{point guards}. 
Here, we consider the point guards.

The Art Gallery problem for guarding simple polygons was proved to be NP-hard for vertex guards by Lee et al. \cite{lee}. This proof was later generalized for point guards \cite{cite:aggarval}. 
In 1987, Ghosh presented approximation algorithms for vertex guards, achieving a ratio of $\mathcal{O}(\log n)$ \cite{cite:Ghosh}, which was improved up to $\mathcal{O}(\log\log\OPT)$ by King and Kirkpatrick in 2011 \cite{cite:k}. It has been conjectured that constant-factor approximation algorithms exist for these problems.
In 2018 Bhattacharya,  S.Kumar Ghosh and S. Prasant Pal presented a constant approximation algorithm for guarding simple polygons by using vertex guards \cite{constant-factor}. 

Assuming integer coordinates and a specific general position on the vertices of $\P$, Bonnet and Miltzow \cite{cite:logopt} presented the first $\mathcal{O}(\log \OPT)$-approximation algorithm for the point guard problem. 
Their result was extended to be one of the most recent works on the point guarding problem in 2020 \cite{cite:2020}.  Assumption 1 (Integer Vertex Representation); vertices are given by integers, represented
in the binary form.
Assumption 2 (General Position Assumption); no three extensions meet in a point of $\P$
which is not a vertex and no three vertices are collinear.
In this paper, we proceed with these assumptions too.

{\bf Our Setting:}
Every guard can see a point if it is directly visible to it or if it is reflected-visible. This is a natural and non-trivial extension of the classical Art Gallery setting. The problem of visibility via reflection has many applications in wireless networks and Computer Graphics, where the signal and the view ray can be reflected on walls. 
There is a considerable literature on geometric optics (such as \cite{optic2} and \cite{optic}), and on the chaotic behavior of a reflecting ray of light or a bouncing billiard ball
(see, e.g., \cite{reflection}, \cite{billiard-2}, and \cite{billiard-4}).
Particularly, regarding the Art Gallery problem, reflection helps in decreasing the number of guards. 
Special cases of this problem have been described by Chao Xu in 2011 \cite{ChaoXu} and by A. Vaezi et al. in 2020 \cite{algorithmica}. 

\begin{pro}[\bf the Art+R Problem]
\label{prob}
Given a simple polygon, $\P$, we intend to find the minimum number of point guards that cover $\P$, considering an extra capability that the single specular reflection is allowed for every edge of $\P$.
\end{pro}

The main idea of solving the above-mentioned problem is to find some areas inside $\P$ in which all points have approximately the same visibility polygon. In other words, if we place a guard $g$ inside such an area, no matter where $g$ lies, it can see approximately the same sub-area of $\P$. We will see what we mean by this approximation later. Such areas are called \emph{guarding-regions}. First, we define some terminologies; then we will see how to compute guarding-regions. When we have the set of all guarding-regions in $\P$, we use the greedy algorithm of the Set Cover problem to find the minimum number of guarding-regions that represent the optimal places for guards that exist in a solution for the problem~\ref{prob}.
\section{Definitions}
This section covers a few terminologies we will use throughout the paper. 

Consider a sub-region $\alpha$ of $\P$, a guard $g$ may be situated in one of the following positions:   1- the guard $g$ does not see any point of $\alpha$, either directly or via reflection. 2- The guard $g$ can see $\alpha$ only partially, either directly or via reflection. 3- The guard $g$ covers $\alpha$  entirely via direct visibility. 4- The guard $g$ covers $\alpha$ entirely via reflection. 5- The guard $g$ covers $\alpha$ entirely; however, considering \emph{both} direct visibility and reflection, i.e. $\alpha$ is partially directly visible and partially reflected-visible to $g$.
In the first and second situations, we consider $\alpha$ as invisible for $g$. That is, for $\alpha$ to be counted as area-visible to $g$, all of its points must be visible to $g$. However, this visibility can be a combination of direct visibility and reflection through reflection-edges. Note that more than one edge may have to be considered regarding reflection;  we need to count on all edges that might add some area to the reflected-visibility of a guard. Moreover, a sub-region $sr$ is area-visible to another sub-region $sr'$, if all points of $sr$ are visible to $sr'$ (either directly or via reflection).
\begin{definition}
\emph{Area-Visibility}: A region $r$ is called area-visible to a source (a point, a segment or a region) if all points of $r$ are visible to all points of $source$.
\end{definition}

\begin{definition}
\label{def:scr}
\emph{Second-order-Convex-Regions}: Each minimal region formed from the following steps is called second-order-convex-regions; the set containing all of them is denoted by $\SCR$.
\begin{enumerate}
    \item Ignore the edges of $\P$. Connect every vertex of $\P$ to each other, and compute
all intersections of all  lines crossing them.
    \item Draw the lines between the previous intersection points with the reflex vertices
of $\P$.
\item  Compute all intersections of all of the above-mentioned lines.
\item  Now trace $\P$ and compute the intersection of the boundary of $\P$ and the
above-mentioned lines.
\end{enumerate}
 \end{definition}
 Figure~\ref{fig.convex-regions} (a) to (g) illustrates an example of decomposing $\P$ into second-order-convex-regions.

  Later, we will see how this decomposition helps us to prove our approximation factor analysis.

  As mentioned before, $\P$ will be decomposed into other smaller special sub-regions called ``guarding-regions”. 
 If we put two guards in different positions inside a guarding-region, they can get area-visible from the same sub-set of second-order-convex-regions. In order to compute guarding-regions, we decompose every second-order-convex-region into guarding-regions.  
For a better definition of the guarding-regions, see the following definition.
 \begin{definition}
  \emph{guarding-regions}:
 Consider a simple polygon $\P$; decompose $\P$ into second-order-convex-regions. Given a query point as a place for a guard $g$ in $\P$, $g$ can make a sub-set of second-order-convex-regions area-visible. Call this set the visible list of $g$ and denote this list by $\VL$. A guarding-region $gr$ is a region $g$ places inside $gr$; if we move $g$ in any arbitrary position inside $gr$, $\VL$ will not change. In other words, the visible list of every point $p \in gr$ is the same.
  \end{definition}

We will estimate the visibility of a guard in an arbitrary position inside a guarding-region $gr$, by a list of second-order-convex-regions area-visible to $gr$. Denote this list by $\VL(gr)$.
Based on this definition, 
for a more straightforward presentation, we may refer to guarding-region's visibility instead of guard visibility. 
In order to compute guarding-regions inside a second-order-convex-region, we need to count on the area-visibility of each second-order-convex-region individually. A sub-region inside each second-order-convex-region that sees another second-order-convex-region completely is called a temporary sub-region; for short, we call each of them a temp-sub-region. See the following definition.

\begin{definition}
\emph{Temp-sub-region:} 
Given the two second-order-convex-regions $scr_{i}$ and $scr_{j}$ ($i\neq j$), suppose there is a sub-region inside $scr_{i}$ that can make $scr_{j}$ area-visible; we call this sub-region a temp-sub-region and denote it by $tmp$.
\end{definition}
Define $t_{i}$ to be the set of all temp-sub-regions inside $scr_{i}$. These temp-sub-regions will be computed by checking the edges of each second-order-convex-region individually. Suppose $scr_{i}$ has $Ed_{i}$ edges. Denote the $j^{th}$ edge of $scr_{i}$ by \emph{$e_{j}(scr_{i})$} $1 \leq j \leq Ed_{i}$. 

Each temp-sub-region eventually will be used to divide second-order-convex-regions into guarding-regions. 

%%%%%%%%%%%%%%%%%%%%%%%%%%%%%%%%%%%%%%%%%%%%%%%%

\section{Algorithm}
\label{full.discription}
Here we present an algorithm that uses three steps to come up with a solution with a $\log n$-approximation factor.

%%%%%%%%%%%%%%%%%%%%%%%%%%%%%%%
%%%%%%%%%%%%%%%%%%%%%%%%%%%%%%%
%%%%%%%%%%%%%%%%%%%%%%%%%%%%%%%

Algorithm \ref{algo.main} illustrates a pseudo-code for our approach. In Step 1, according to Definition~\ref{algo.main}, we decompose $\P$ into second-order-convex-regions, with $|\SCR|$ indicating the number of all second-order-convex-regions. Step 2 contains three main $\for$ loops that correspond to each second-order-convex-region; we check their edges one by one to find temp-sub-regions. Consider $scr_{i}$ with $Ed_{i}$ edges. Any guard $g$ inside $scr_{i}$ may see other second-order-convex-regions ($scr_{j}$ $i\neq j$) only through
    $scr_{i}$'s edges. 
    Step 2 uses a procedure called $Find$ to compute temp-sub-regions corresponding to an edge of a given second-order-convex-region. All in all, Step 2 computes temp-sub-regions and puts them in a set denoted by $t$. The $Find$ procedure is the only part using reflection.

Step 3 receives all temp-sub-regions and decomposes them into guarding-regions. This step uses another procedure called $Decompose$. This procedure is given a second-order-convex-region as a parameter, splitting that second-order-convex-region into guarding-regions. The given second-order-convex-region, denoted by $scr_{i}$, will be decomposed into a set of guarding-regions called $s_{i}$. Denote $\S$ to be the final set that contains all guarding-regions, and let $|\S|$ denote the cardinality of $\S$. The algorithm computes every $s_i$ and adds it to $\S$ ({\small $1\leq i \leq {\tiny |\SCR|}$}).
As demonstrated by line $17$, $\S$ contains ordered-pairs ($\mathit{<scr,\VL(sr)>}$) of all guarding-regions. Every guarding-region owns a visible-list. Denote the cardinality of a visible-list by $|\VL|$. In Step 4, the set $\F$ is the final solution of the algorithm, which is a sub-set of $\S$. Denote the cardinality of $\F$ by $|\F|$.  To cover $\P$, all second-order-convex-regions must get covered by the guarding-regions chosen in the final set $\F$. The problem of selecting the minimum number of sub-sets (guarding-regions) in $\S$ can be seen as an instance of the well-known Set Cover problem. The solution leads to the final set $\F$.

%%%%%%%%%%%%%%%%%%%%%%%%%%%%%%%%%%%%%%
\begin{theorem}
\label{theo.proof.opt}
There is a $\log (n)$-approximation algorithm for the Art+R problem.% that is the best possible solution.
\end{theorem}
\begin{proof}
By running Algorithm~\ref{algo.main}, we can convert the Art+R problem to an instance of the Set Cover problem. Using a greedy approach to solve an instance of the Set Cover problem, we can obtain the $\log(n)$ factor for the approximation. Subsection~\ref{sec:approximation-analysis} deals with the approximation analysis of the algorithm in details.
\end{proof}

%%%%%%%%%%%%%%%%%%%%%%%%%%%%%%%%%%%%%
%%%%%%%%%%%%%%%%%%%%%%%%%%%%%%%%%%%%%
%%%%%%%%%%%%%%%%%%%%%%%%%%%%%%%%%%%%%
\begin{algorithm}[t]
\caption{Main Procedure}
\label{algo.main}
\begin{algorithmic}[1]
\Procedure{Art+R($\P$)}{}
\State define $\longleftarrow$ = Add an element to a set/list.
\State define $tsr_{i}(scr_{j})$ = The temp-sub-region in the $i^{th}$ second-order-convex-region, where $scr_{j}$ is area-visible to $tsr_{i}(scr_{j})$.

//{*\small { \bf Step 1}}

\State decompose $\P$ into second-order-convex-regions and put them $\SCR$.

//{*\small { \bf Step 2}}

\For{ i= 1 ; $i \leq  |\SCR|$; i++ } //{\small * for every second-order-convex-region}
    \For{ k=1 ; $k \leq  |Ed_{i}|$; k++ } //{\small * pick $scr_{i}$'s edges one by one}
               \For{ j=1 ; $j \leq  |\SCR|$; j++ } //{\small * find each $tsr$ corresponding to every $scr_{j}$ inside each $scr_{i}$ }
                  \State List $L$ = Find($ed_{ik}$,$scr_{j}$) //{\small * $ed_{ik}$ is the $k$th edge of $scr_{i}$ }
                  \For{ each $tcr \in L$} 
                  \State   {\bf if} $tsr_{i}(scr_{j})$ sees $scr_{j}$ {\bf do}
                  \State    $\ \ \ $ $\VL(temp_{i}(scr_{j}))$ $\longleftarrow$ $scr_{j}$
                  \State    $\ \ \ $ $t_{i}$ $\longleftarrow$ $tsr_{i}(scr_{j})$
                  \EndFor
               \EndFor
    \EndFor
\EndFor
//{*\small { \bf Step 3}}

//{\small \emph{*Extracting guarding-regions from the temp-sub-regions in each $scr$:}}
\For{ i= 1 ; $i \leq  |\SCR|$; i++}
        \State Decompose($scr_{i}$)
        \State $s_{i}$ $\longleftarrow$ \emph{$<$the result regions, corresponding visible-lists$>$};
        \For{ each~$gr$~in $s_{i}$ }
                 \State $\S$ $\longleftarrow$ $\mathit{<gr,\VL_{gr}>}$;
            \EndFor
\EndFor
//{*\small { \bf Step 4}}
\State $\F$ $\longleftarrow$ best possible sub-set of $\S$ that covers $\P$.
\State Return $\F$
\EndProcedure
\end{algorithmic}
%\captio}
\end{algorithm} 
%%%%%%%%%%%%%%%%%%%%%%%%%%%%%%%%%%%%%%%%%%%%%%%%%
\begin{figure}[thp]%figure8
\begin{center}
\includegraphics[scale=0.6]{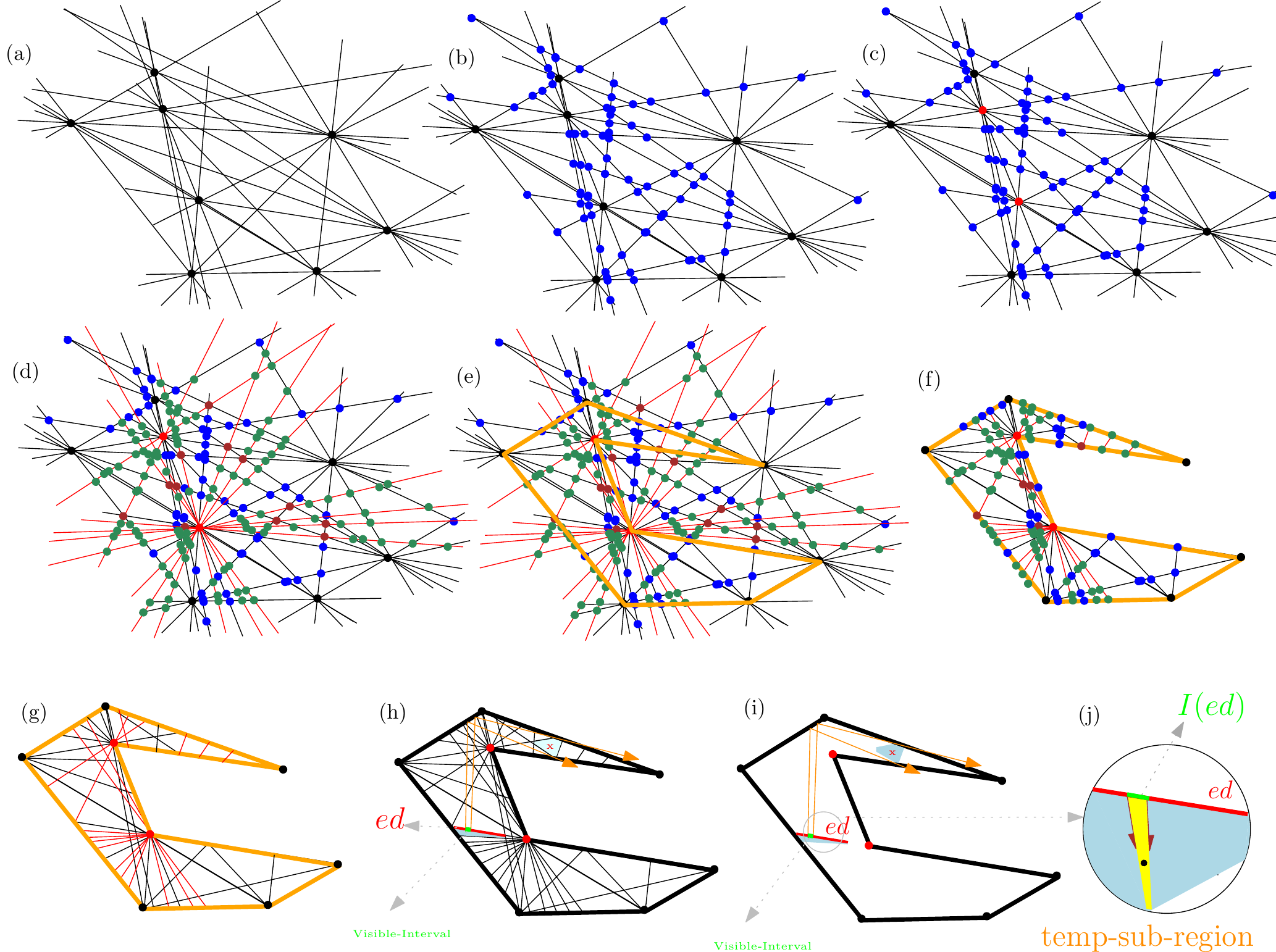}
\caption{ Figures (a) and (b) illustrate the first step of Algorithm~\ref{algo.main}. Figure (c) shows the reflex vertices. Figure (d) also illustrates that reflex vertices must get connected to the intersection points. Although not all lines are illustrated, figure (d) shows most of them. Figure (e) shows the boundary of the polygon. Figure (f) points out the interior of the polygon. Meanwhile, figure (g) illustrates the lines without the intersection points. Figure (h) illustrates two second-order-convex-regions in blue. The one denoted by $x$ is the target second-order-convex-region. We intend to find a temp-sub-region inside the other one (the $source$) to make $x$ area-visible. One edge of the source second-order-convex-region is denoted by $ed$. We use the backward rays making $x$ area-visible to $ed$. In figure (i), the backward half-lines make a temp-sub-region inside the specified second-order-convex-region. Every point on the green interval $I(ed)$ can make $x$ area-visible. In fact, $I(ed)$ is the strong-visible part of $ed$ for making $x$ area-visible. Figure (j) illustrates the temp-sub-region inside the source in yellow.}
\label{fig.convex-regions}
\end{center}
\end{figure}  

%%%%%%%%%%%%%%%%%%%%%%%%%%%%%%%%%%%%%%%%%%%%%%%%%
%%%%%%%%%%%%%%%%%%%%%%%%%%%%%%%%%%%%%%%%%%%%%%%%%

\subsection{{\bf procedure Find (a segment, a second-order-convex-region)}}
\label{subsec.algo.convex}
   Consider $scr_{i}$ with $Ed_{i}$ edges. Any guard $g$ inside $scr_{i}$ may see other second-order-convex-regions ($scr_{j}$ $i\neq j$) only through
    $scr_{i}$'s edges. 
   Line 6 of the Algorithm~\ref{algo.main}, in turn, checks all such edges for $scr_{i}$.
    Consider an edge $ed$ of $scr_{i}$. Suppose $g$ can see $scr_{j}$ $i\neq j$ through $ed$.
    There must be an interval $I(ed)$ of $ed$ from which the rays from $g$ to $scr_j$ cross $I(ed)$. 
    The interval $I(ed)$ is the strong-visible part of $ed$, whereas $scr_{j}$ is area-visible to $I(ed)$. i.e., every point on $I(ed)$ can see $scr_{j}$. The reason this interval is important for us is because it determines an area inside $scr_{i}$ in which every point can make $scr_{j}$ area-visible. Using the endpoints of $I(ed)$, we can find the rays that can make $scr_{j}$ area-visible. By extending these rays backward inside $scr_{i}$, we can find a temp-sub-region $tsr$ inside $scr_{i}$.
    ~Both $scr_{i}$ and $scr_{j}$ must be added to the visible-list of $tsr$ (see Figure~\ref{fig.convex-regions}).

In this subsection, we intend to \emph{find} intervals on a segment $ed$ as a \emph{viewer} that can completely see a second-order-convex-region called $target$. Then, we will compute the temp-sub-regions corresponding to $ed$ inside the given second-order-convex-region.

Suppose that $ed$ is a given edge of a second-order-convex-region called $source$. 
There are three cases: 1- an interval (a point or segment) on $ed$ can see the entire target {\bf directly}. 2- An interval of $ed$ can see the entire target {\bf via reflection}. 3- Interval(s) of $ed$ can see the entire target {\bf considering both direct and reflected-visibility}. (In this case, there might be more than one interval on $ed$ that can make the given second-order-convex-region area-visible.)

\begin{lemma}
\label{lem:find}
Given a second-order-convex-region denoted as $target$ and a segment $ed$, which is an edge of another second-order-convex-region denoted by $source$, we can find temp-sub-regions in $source$ that can make the target area-visible. This takes $\mathcal{O}(n|\SCR|\log n)$ time complexity, where $|\SCR|$ indicates the complexity of a second-order-convex-region.
\end{lemma}
\begin{proof}
Subsection~\ref{apen:subsec.algo.convex} in Appendix deals with the proof in details. 
\end{proof}

%%%%%%%%%%%%%%%%%%%%%%%%%%%%%%%%%%%%%%%%%%%%%%%%%%%%%%%%%%%%%%%%%%%%%%%%%%%%%%%%%%%%%%%%%%%%%%%%%%%%%%%%%%%%%%%%%%%%%%%%%%%%%%%%%%%%%%%%%%%%%%%%%%%
\subsection{Procedure {\bf Decompose(second-order-convex-region)}}
\label{sub.algo.intersect-summary}

This procedure is intended to decompose a given second-order-convex-region ($scr$) into guarding-regions. The complete version of this subsection is explained in Appendix subsection~\ref{sub.algo.intersect}.

According to Algorithm~\ref{algo.main}, when we reach this procedure, every second-order-convex-region is divided into some temp-sub-regions that probably share common areas. Line 12 counts every second-order-convex-region; line 13 calls Decompose to obtain guarding-regions inside a given second-order-convex-region. The Decompose procedure computes guarding-regions from the temp-sub-regions, so that the union of the guarding-regions inside a second-order-convex-region equals the given second-order-convex-region;  also, guarding-regions do not share any area (see Figure~$\ref{fig.decopposing.cases}$).

\begin{figure}[tp]%figure2
\begin{center}
\includegraphics[scale=0.4]{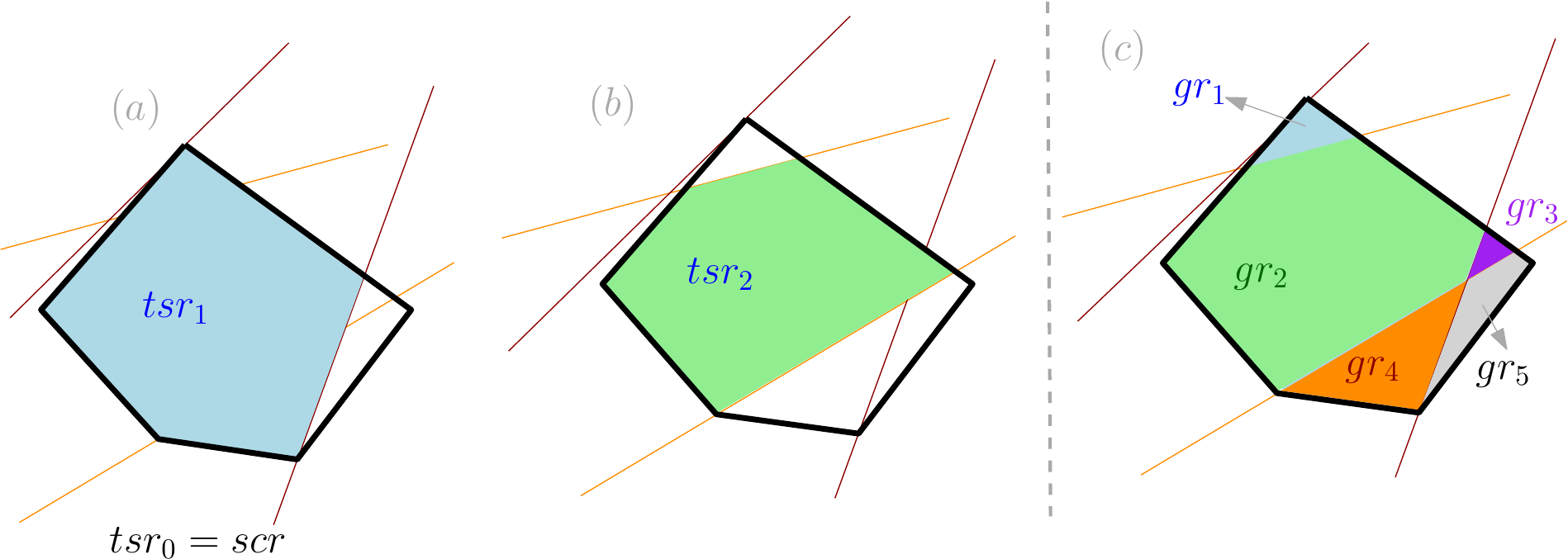}
\caption{This figure illustrates what decomposing a second-order-convex-region means. Three second-order-convex-regions are considered in this figure. The one denoted by $scr$, and two others are referred to as $scr_{1}$ and $scr_{2}$. Figure $(a)$ illustrates a temp-sub-region ($tsr_{1}$) in blue. Any guard inside $tsr_{1}$ can make $scr_{1}$ area-visible. Figure $(b)$ illustrates that $tsr_{2}$ can make $scr_{2}$ area-visible. Finally, figure~$(c)$ illustrates the  decomposition of this  temp-sub-region into four guarding-regions. 
In figure $(c) $, consider $gr_{4}$ as an example. The guarding-region of $gr_{4}$ is illustrated in orange, making $scr$ and $scr_{1}$ area-visible, but not $scr_{2}$.}
\label{fig.decopposing.cases}
\end{center}
\end{figure}

A temp-sub-region $tsr$ is obtained from the area between two half-lines (backward rays). 
We aim to decompose every second-order-convex-regions, $scr$, using their intersection with the area between these half-lines. 
Denote the number of temp-sub-regions inside $scr$, including $scr$, by $|scr|$, indicating the complexity of $scr$.
For each temp-sub-region $tsr$ in $scr$, except for $scr$ itself, denote a starting half-line by $shl(tsr)$, and an ending half-line by $ehl(tsr)$. So, inside $scr$, a temp-sub-region $tsr$ is between two half-lines $shl(tsr)$ and $ehl(tsr)$. We intend to sweep on $scr$, regardless of the direction we move on; for a simpler presentation, we refer to these half-lines $shl(tsr)$ and $ehl(tsr)$.
From the previous steps of Algorithm~\ref{algo.main}, for every two starting and ending half-lines, it is already specified what second-order-convex-regions ($\neq scr$) are visible to the points between any two half-lines $shl$ and $ehl$.

\begin{lemma}
\label{proof:decompose}
A second-order-convex-region can be decomposed into guarding-regions in $\mathcal{O}(|scr|\log(|scr|))$, where $|scr|$ denotes the complexity of a second-order-convex-region.
\end{lemma}
\begin{proof}
To decompose $scr$, we use a sweep-line denoted by $\cal{SL}$. We start sweeping parallel to an edge of $scr$. Various events may appear, $\cal{SL}$ checks every possibility and decide its action in each step. Appendix subsection~\ref{sub.algo.intersect} is concerned with the decomposition procedure. As the complexity of a second-order-convex-region is  $\mathcal{O}(|scr|)$, the sweeping approach may need $\mathcal{O}(|scr|\log(|scr|))$ time.
\end{proof}

\subsection{Analysis}
\label{analyze}
 Figure~\ref{fig.example.n} is an example showing that single specular reflection lowers the number of guards from $\Theta(n)$ to $1$ guard. So, reflection can help to eliminate a large number of guards. Of course, mirrors are cheaper than cameras. 

\begin{figure}[]%figure2
\begin{center}
\includegraphics[scale=0.6]{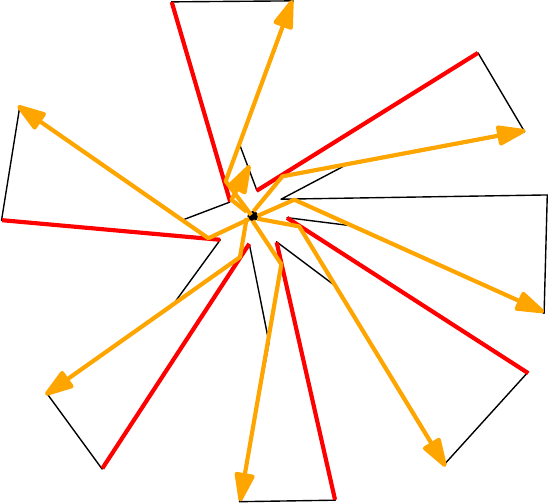}
\caption{This figure illustrates a situation where a single guard is required if we use reflection-edges; $\Theta(n)$ guards are required if we do not consider reflection. Red segments illustrate the reflected-edges.}
\label{fig.example.n}
\end{center}
\end{figure}

The rest of this section covers the analysis of Algorithm~\ref{algo.main}. First, note that due to line $5-12$ of the algorithm, we know that the guards chosen by the algorithm surely cover $\P$; so, the algorithm provides a feasible solution. That is because when $i=j$, every second-order-convex-region will be added to $\S$. Although each second-order-convex-region $scr$ is decomposed to temp-sub-regions, each of the resulting temp-sub-regions has $scr$ in its visible-list. So, all  second-order-convex-regions and consequently, $\P$ will be covered.

%%%%%%%%%%%%%%%%%%%%%%%%%%%%%%%%%%%%%%%%%%%%%%%%%

\begin{lemma}
\label{lem: polynomial}
Algorithm~\ref{algo.main} takes polynomial time complexity.
\end{lemma}
\begin{proof}
A simple analysis reveals that the algorithm works in the polynomial time: 
$m$ lines can have at most $m^{2}$ intersections. The polygon has $n$ vertices. First, we find the lines formed by joining vertices. We have $\mathcal{O}(n^2)$ lines, and $\mathcal{O}(n^{4})$ intersection-points. We connect these intersection-points to $\mathcal{O}(n)$ reflex-points, creating $\mathcal{O}(n^{5})$ lines. So, there are $\mathcal{O}(n^{10}$ intersection points, providing the final list of second-order-convex-regions. The list of second-order-convex-regions contains $\mathcal{O}(n^{10})$ elements. In each second-order-convex-region, there are at most $\mathcal{O}(n^{10})$ temp-sub-regions. So, in the worst case, there could be $\mathcal{O}(n^{20})$ guarding-regions in $\S$.

A second-order-convex-region has $\mathcal{O}(n^{10})$ edges. That is, $|\SCR| \in \mathcal{O}(n^{10})$.
~A sweeping line on temp-sub-regions in a second-order-convex-region takes at most $\mathcal{O}(n^{10}\log(n))$ time. As the algorithm presented in \cite{tcs} needs $\mathcal{O}(n)$ time, the Find procedure takes $\mathcal{O}(n^{11}\log(n))$ time to find temp-sub-regions inside a second-order-convex-region.
A decomposition process is one time sweeping on a second-order-convex-region, resulting in splitting the surface of that second-order-convex-region into several guarding-regions; so, this procedure does not take more time than $\mathcal{O}(n^{10})\log (n)$.

All in all, considering Algorithm~\ref{algo.main} lines of five to twelve takes $\mathcal{O}(n^{41}\log n)$ time in a worst-case analysis.
Moreover, lines thirteen to seventeen take $\mathcal{O}(n^{20}\log n)$ time.

So, the algorithm works in the polynomial time, and this is what matters for us. This is an upper bound; however, whether the complexity can be improved remains open.
Decreasing the complexity of the second-order-convex-regions has a great impact on the complexity of the algorithm. As we do not count on partial visibility; any change must be matched with the approximation proof to ensure it does not affect the analysis. 
\end{proof}

\subsubsection{Approximation Analysis}
\label{sec:approximation-analysis}
This section proves that the approximation factor of Algorithm~\ref{algo.main} is $\log n$. As we already know that this approximation is the best known considering the standard \emph{point} guarding problem in the worst case, we can conclude that it is a nice approximation ratio for the generalized problem.  

%%%%%%%%%%%%%%%%%%%%%%%%%%%%%%%%%%%%%%%%%%
%%%%%%%%%%%%%%%%%%%%%%%%%%%%%%%%%%%%%

 Theorem~\ref{theo.proof.opt} claims that
Algorithm~\ref{algo.main} provides a $\log n$-approximation for the AG+R problem. In the following, we prove this theorem.

\begin{proof}
\emph{All through the proof we only count on direct visibility;
however, reflection may help the visibility of a guard, making the situation even better.}

The set $\S$ contains $|\S|$ guarding-regions every one of which, say $gr$, has a list of second-order-convex-regions (visible-list) $\VL(gr)$; so, if we put a guard $g$ in an arbitrary position in $gr$, $g$ can see all second-order-convex-regions in $\VL(gr)$ entirely. However, the guard $g$ might see some other second-order-convex-regions partially. Since Algorithm~\ref{algo.main} only counts on the second-order-convex-regions that are entirely visible (area-visible), it does not add the partially visible part of the  second-order-convex-region to a visible list of a guarding-region.
These partially visible areas cause a difference between the solution we provide by Algorithm~\ref{algo.main} and the optimal one. That is because one or more guards in the optimal solution might see some second-order-convex-regions $scr$s partially; the rest of $scr$s gets visible by other guards.

To prove Theorem~\ref{theo.proof.opt}, we have two steps; first, considering that we have a set $\S$, we need to select the minimum number of sub-sets (guarding-regions) from that to cover all second-order-convex-regions. The first step is an instance of the Set Cover problem, which has a $\log (|\SCR|)$-approximation algorithm. 
To organize our instance of the Set Cover problem, we consider $\SCR$ as the set of items. So, all second-order-convex-regions should be covered. The visible-list of every guarding-region corresponds to a subset of our instance problem. The minimum number of guarding-regions whose visible-list covers all second-order-convex-regions is the solution of this instance of the Set Cover problem.

Secondly, we need to know how well $\S$ is.
To measure how good $\S$ is, we need to compare the optimal solution of the Art+R problem, denoted by $\OPT^{*}$, and the optimal solution we have from $\S$, denoted by $\OPT(\S)$. 
Denote $\OPT^{*}$ as the set of guards in the real optimal solution. Suppose there are $|\OPT^{*}|$ guards in $\OPT^{*}$. Denote $|\OPT(S)|$ as the number of guarding-regions in $\OPT(\S)$.
In the following lemma, we prove that $|\OPT(\S)|\leq$ $\frac{3}{2} |\OPT^{*}|$. So, Algorithm~\ref{algo.main} is a $O(\log (|SCR|))$-approximation for the Art+R problem.

Since from Lemma~\ref{lem:polynomial}, $|\SCR|$ is polynomial, we can conclude that Algorithm~\ref{algo.main} provides a $\log (n)$-approximation factor for the Art+R problem, where $n$ indicates the complexity of $\P$.
\begin{lemma}
\label{lem:2-app}
$|\OPT(\S)|$ $\leq \frac{3}{2}$ $|\OPT^{*}|$
\end{lemma}

\begin{proof}
Every guard $g_{opt}$ in $\OPT^{*}$ covers a sub-area of $\P$ denoted by $sa(g_{opt})$, which is, in fact, the visibility polygon of $g_{opt}$.
Pick an arbitrary point $p$ in a guarding-region $gr$. Set $g_{s}$ to be in the position specified by $p$. The guard $g_{s}$ can see a sub-area of $\P$ denoted by $sa(g_{s})$.
The proposed algorithm tries to estimate $sa(g_{s})$ with a sub-set of second-order-convex-regions. This sub-set is determined by the visible-list of a guarding-region ($\VL(gr)$). 
Choosing a guarding-region in $\S$ corresponds to selecting a position for a guard.

Clearly, the guards in $\OPT^{*}$ also cover all  second-order-convex-regions.
Remember that for every guard, Algorithm~\ref{algo.main} counts only on the complete visibility of second-order-convex-regions. If a subset of guards in $\OPT^{*}$ sees a set of second-order-convex-regions via complete visibility, and a larger set in $\OPT{\S}$ covers the same second-order-convex-regions, then we can choose the same guards form $\OPT^{*}$ to be in $\OPT(\S)$. Otherwise, $|\OPT^{*}|$ and $|\OPT(\S)|$ would have the same number of guards.

So, we only have to care about situations where at least one second-order-convex-region is covered by some previously chosen guards in $\OPT^{*}$, so that this second-order-convex-region is covered by the union of the partial visibility of some guards in $\OPT^{*}$. In $\OPT(\S)$, we have to count on extra guards to cover partial visible second-order-convex-regions completely.
We will see that since reflex vertices are connected to the intersection points from the criss-cross of the lines of the polygon, situations in which a few guards in $\OPT^{*}$ can see many second-order-convex-regions partially cannot happen. However, depending on the positions reflex vertices and guards may happen to choose, a few cases are investigated for the proof to be complete. See Appendix~\ref{sub;proof} for the details of the proof.
\end{proof}
\end{proof}

%%%%%%%%%%%%%%%%%%%%%%%%%%%%%%%%%%%%%%%%%%%%%%%%%%%%%%%%%%%%%%%%%%%%%%%%%%%%%%%%%%%%%%%%%%%%%%%%%%%%%%%%%%%%%%%%%%%%%%%%%%%%%%%%%%%%%%%%%%%%%%%%%%%%%%%%%%%%%%%%%%%%%%%%%%%%%%%%%%%%%%%%%%%%%%%%%%%%%%%%%%%%%%%%%%%%%

\section{Discussion}
  The paper generalizes the classical Art Gallery problem by allowing “specular reflections”, considering them only once. Similar to the classical version, this one also has a $\log$-space approximation, by applying a transformation to the well-known Set Cover problem.
  
Since some years ago, scientists have considered reflection in their problem-solving. 
In this article and some previous works, we have considered reflection as an extra capability to enhance other problems. 
The question of whether $\mathcal{O}(\log n)$ is the best possible approximation factor for the Art+R problem, however, remains open. 
However, we know that $\log n$  is the best-known approximation factor for the normal \emph{point} guarding problem in the worst case.

Even though multiple reflections may seem more interesting from a theoretical point of view, in practice, each mirror lowers the quality of the pictures, especially when dealing with a large gallery.  From a more practical point of view, we might tend to consider placing only a few (small) mirrors somewhere (as mirrors are expensive and we use a building as an Art Gallery) to reduce the number of guards. This algorithm helps us to determine which edge and which part of that edge can provide a larger view. If a user has a fixed camera, they can use the algorithm proposed by \cite{tcs}. Otherwise, the algorithm presented in this paper specifies the positions of guards and the mirrors together.

 Although the main contribution of this paper is to consider reflection, this algorithm also works for the standard Art Gallery problem. This can be done by changing the Find function only to consider direct visibility. The modified algorithm converts the point guarding problem to a version of the Set Cover problem; however, the sub-sets (the guarding-regions) are not independent, and this is not a general case of the Set Cover. We intend to find the relations between the guarding-regions and provide a better approximation factor for the standard point guarding problem in the future work.

%---------------------------- Appendix -------------------------------

\newpage
\section*{Appendix}
\label{appendix}
\subsection{{\bf Procedure Find (a segment, a second-order-convex-region)}}
\label{apen:subsec.algo.convex}
This section deals with the various cases of the Find procedures in details.

1) \emph{An interval on $ed$ can see the entire target {\bf directly}}:

We need to find the largest interval $I(ed)$ that can make the target area-visible. To do this, since we know that the target is a convex-region, we only need to make sure that all vertices of the target are visible for all points of $I(ed)$. The intersection of the visibility of all target's vertices points out an interval on $ed$ that is denoted as $I(ed)$. 
In other words, locate the two extreme points of $ed$ visible from each vertex of $target$. Then, find their common intersection.

We use $I(ed)$ to find a temp-sub-region, $tsr$, that can see the target. We add $tsr$ to \emph{the output list} of the Find procedure.
    As $I(ed)$ may get larger by the help of reflection, we have to check other cases too.
    However, this case obviously only takes $\mathcal{O}(|\SCR|)$.

 2)   \emph{An interval of $ed$ can see $target$ entirely {\bf via reflection}:}
 
We use the algorithm in subsection $5.3$ of \cite{tcs}. This algorithm takes a segment as an input denoted by $\seg{xy}$;  we need to choose two specific rays from four ones to make sure that a target gets completely visible to the viewer. The chosen rays specify two half-lines denoted by $L_{y}$ and $L_{x}$; starting from the endpoints of $\seg{xy}$, they will start from a reflector and point to the target. The chosen half-lines are denoted by numbers $1$ or $2$.  %%; to use that algorithm, we 

To use that algorithm, set $ed$ to be $\seg{xy}$ as the viewer, and use $L_{2,y}$ and $L_{1,x}$ half-lines to make sure that $target$ is \emph{completely} reflected-visible to $I(ed)$. These half-lines must be released from $ed$ to reach $target$ after reflecting from a reflection-edge $e$ in the middle. The opposite direction from $e$ to $ed$ makes two half-lines to intersect $ed$ on the endpoints of $I(ed)$. 
There could be other such edges like $e$ that are used with $ed$ to make $target$ reflected-visible; however, we need to take all such edges into account and find the largest $I(ed)$. This can be done by the mentioned algorithm, as presented by \cite{tcs}. After finding the correct half-lines that make the endpoints of $I(ed)$, we extend these half-lines inside the source. The intersection between these two half-lines and the source creates a temp-sub-region $tsr$. By using the binary search, we can compute this intersection in $\mathcal{O}(\log (|source|))$ time, where $|source|$ denotes the complexity of the source. Note that any point in $tsr$ can make both target and source area-visible.

The algorithm we used from \cite{tcs} could be only useful if we were looking for a segment to get visible to a source segment; however, here our target is a convex-region, not a segment. Suppose that $target$ has $k$ edges; to be sure the whole $target$ is reflected-visible to the given segment (the viewer), we have to take each edge of $target$ into account individually. 

Consider all edges of the $target$ second-order-convex-region. To compute the final result, we compute the intersection between every two temp-sub-regions corresponding to two edges of the $target$ region in a step by step manner to find the final temp-sub-region.

Denote the complexity of the target by $|target|$. For a specific edge $ed$, to compute a temp-sub-region $tsr$,  it takes $\mathcal{O}(|target|)$ time to compute $I(ed)$; and it takes $\mathcal{O}(\log |source|)$ to compute the boundary of $tsr$.  Thus, it takes $\mathcal{O}(|target|+\log |source|)$ for each run of the Find procedure. Later, we will see that that the complexity of a second-order-convex-region is polynomial in terms of $n$; so, the Find procedure takes $\mathcal{O}(|target|)$ time (considering this case).

Note that the algorithm presented in \cite{tcs} considers all  edges of $\P$ as potential reflection-edges. So, if the $target$ gets covered by different reflection-edges, we surely consider all such reflection-edges. 

 3) \emph{Points on an interval $I(ed)$ on $ed$ might see a part of $target$ directly, and the rest of the $target$  is done via reflection.}

To compute these types of intervals on $ed$, we use the following approach. 
The mentioned approach, as represented below, only considers a source segment $ed$ and a target \emph{segment} denoted by $\seg{uw}$. We need to run this approach on all edges of the $target$ second-order-convex-region and consider the intersection of the computed intervals on $ed$. As a result, we know that every point of the final intervals on $ed$ can see the entire $target$ region. Note that in this case, as various mirror-edges can create different reflected-visible parts on the target by considering direct visibility, there can be various intervals on the source that can make the target area-visible. So, there might be more than one interval on $ed$ that can make the entire area visible.

This third case deals with some special situations where there are some points on the target, such that each one is a place where direct visibility meets the reflected visibility.
For example, see Figure~\ref{fig.sweep-ref}(c). This figure illustrates the target by $\seg{uw}$. The point $p_{1}$ on $ed$ sees the target completely. If we move a little bit on $ed$ from $p_{1}$ and go to $p_{0}$,  the direct visibility gets disconnected from the $e_{1}$-mirror-visible part. We need to compute points like $p_{1}$ or its image on $\seg{uw}$, which is denoted by $q$, to know what interval on $ed$ can see the entire target. This is because $p_{1}$ and such points determine an endpoint for these intervals. Note that computing $p_{1}$ or $q$ is equivalent. That is because $q$ is the projection of $p_{1}$ on the endpoint of the mirror-edge.

Using the algorithm presented in subsection 5.4 of the article \cite{tcs}, we can compute which interval of $ed$ can see which part of $\seg{uw}$ in a weak visible way. To explain more, see Figure~\ref{fig.sweep-ref}(e). We can find out that the points on $ed$ form $p_{0}$ till $p_{3}$ can see at least one point on $\seg{uw}$ from $w$ to $f$ using $e_{1}$.  
However, if we compute direct visibility in $p_{3}$, we will see that direct visibility and $e_{1}$-mirror-visibility share a part on $\seg{uw}$. This part has been illustrated in Figure~\ref{fig.sweep-ref}(e) in green. As we move on $ed$ from $p_{3}$ to $p_{1}$, the direct visible interval gets shorter directed to $u$; also, $e_{1}$-mirror-visible part gets shorter directed to $w$ till the endpoints of these intervals meet each other in $q$.

We have the endpoints of intervals such as $\seg{p_{0}p_{3}}$ form the algorithm previously mentioned in the article \cite{tcs}, but we do not have points such as $p_{1}$; so, we should compute these types of points. The reason we have to compute these points is because they determine the endpoints of visibility intervals on the source. To have a better presentation, let call the points on the source ($ed$); the image on the target is a point where direct visibility and mirror-visibility are merged by $dmvm$. Note that for every mirror-edge, there can be one $dmvm$. Also, a $dmvm$ point might not see the target completely, but it may see the target by getting help from other mirrors. So, we have to compute the $dmvm$ points no matter whether it can see the target entirely or not. We will check that later.
\begin{definition}
A $dmvm$ point on the source is a point where the direct visibility of $dmvm$ and its mirror-visibility on an endpoint of a predetermined mirror-edge are merged on a target which is a segment. In other words, there is a segment as a source and a segment as a target; if we stand on the $dmvm$ point on the source, we can see a direct visible part and a mirror-visible part for a mirror-edge, say $e$. The $e$-mirror-visible part and the direct visible part share one of their endpoints. This endpoint is the projection of $dmvm$ on the target.
\end{definition}
\begin{figure}[t]%figure8
\begin{center}
\includegraphics[scale=0.5]{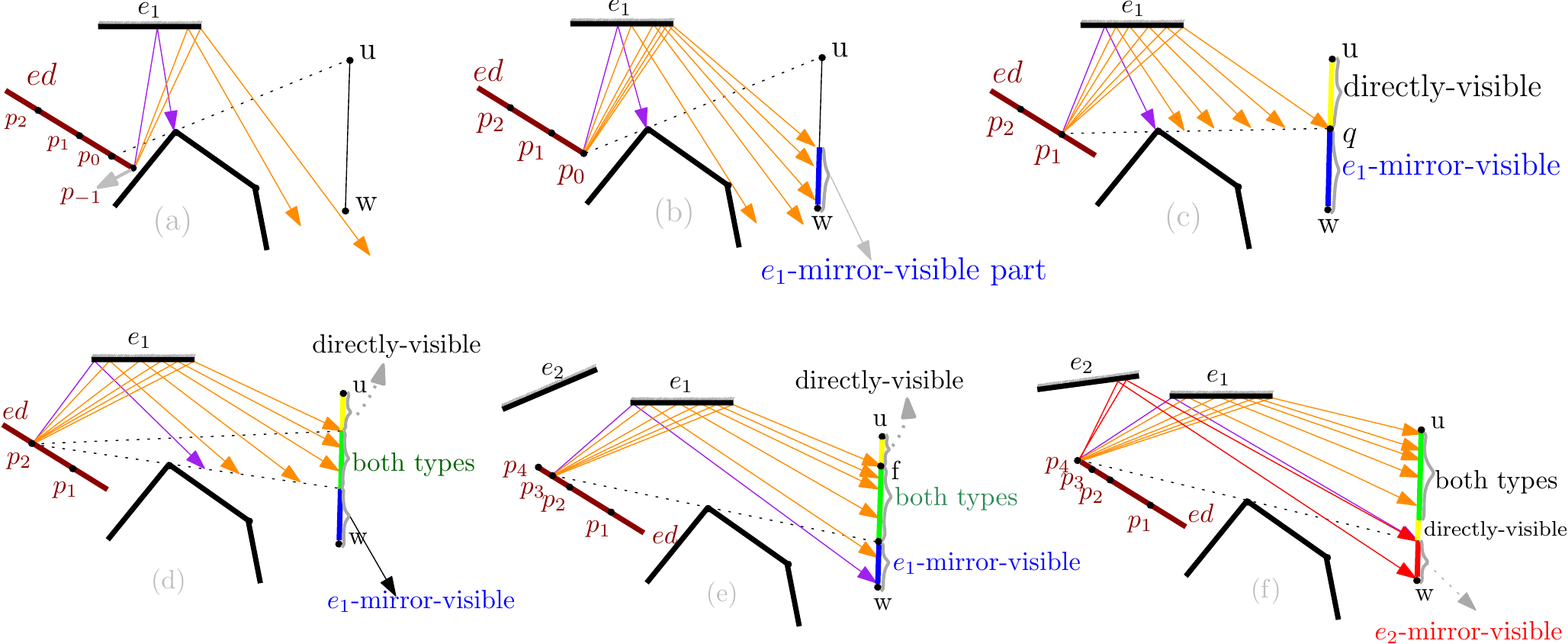}
\caption{Start sweeping from $p_{-1}$ upward on $ed$. a) $p_{-1}$ cannot see any point on $\seg{uw}$. b) However, if we move a little bit, the points on $ed$ can see $\seg{uw}$ partially through $e$ via reflection, but nothing is visible directly. c) The direct-visibility starts to begin on $\seg{uw}$, while we sweep upward on $ed$ till the part  of $ed$ that makes direct-visibility reaches the interval on $ed$ which makes $e$-reflected-visibility on $\seg{uw}$. The point $p_{1}$ is where the points on $ed$ start to see $\seg{uw}$ entirely. This complete visibility involves both reflected-visibility and direct visibility. d) The point $p_{2}$ also lies on the interval of $ed$ that can see the entire $\seg{uw}$ segment via both types of visibility. However, $p_{2}$ can see a common part either directly or via reflection. As we move upward on $ed$, the direct visible interval of $\seg{uw}$ gets larger from $u$ to $w$.  The direct visibility may disappear eventually. Note that in the other direction not used here, 
the reflected-visibility may disappear from a point, while we sweep on the source segment. So, there may be such an event for a point where the reflected-visibility disappears. e) The point $p_{3}$ is where the points on $ed$ stop seeing $\seg{uw}$ via reflection-edge $e_{1}$.
 f) The point $p_{4}$ is a point that can see $\seg{uw}$ using two reflection-edges, or via $e_{2}$, as well as direct visibility.}
\label{fig.sweep-ref}
\end{center}
\end{figure}

\begin{lemma}
For a given segment as a source and a given segment as a target, we can compute all $dmvm$ points with regard to all mirror-edges in linear time, considering the complexity of $\P$. 
\end{lemma}

\begin{proof}
Without the loss of generality, we only deal with one mirror-edge and direct visibility. 
Now, see Figure~\ref{fig:proof}. Without loss of generality, we deal with one case in Figure~\ref{fig:proof}, and we can compute the $dmvm$ point denoted as $p_{1}$ through algebraic equations. 
Since from subsection 5.4 of the article \cite{tcs}, we have $p_{3}$, we can compute the intersection of the ray coming out of $p_{3}$ and reflecting on $v_{2}$. This point is denoted by $f$. Note that as direct visibility is blocked partially, there must be a dominant reflex vertex here; it is denoted by $rf$. The intersection of the line passing  through $p_{3}$ and $rf$ gives us the end of  the directly visible part seen by $p_{3}$ on $\seg{uw}$.  This point is denoted by $h$. The position at which the direct visibility meets mirror-visibility on $\seg{uw}$ is denoted by $q$. In fact, $q$ is the projection of $p_{1}$ on $\seg{uw}$. By definition, $q$ is seen by both direct visibility and mirror-visibility. The distance between $p_{3}$ and $p_{1}$ is called $d$. The only thing we need to compute is $d$. This is because by knowing $d$, we can compute the coordinates of $p_{1}$ and $q$. As mentioned, $q$ must be somewhere between $f$ and $h$. So, set $d$ as a variable and compute the coordinates of $p_{1}$ from the coordinates of $p_{3}$, and in terms of $d$. The line crossing $p_{1}$ on $rf$ must intersect $q$. So, compute the coordinates of $q$ in terms of $d$ too. Also, checking mirror-visibility and knowing the fact that a ray reflecting on $v_{2}$ must cross $q$, we can come up with another equation. By solving these equations, we can compute $d$. Thus, $p_{1}$ is obtained.

\begin{figure}[t]%figure8
\begin{center}
\includegraphics[scale=0.7]{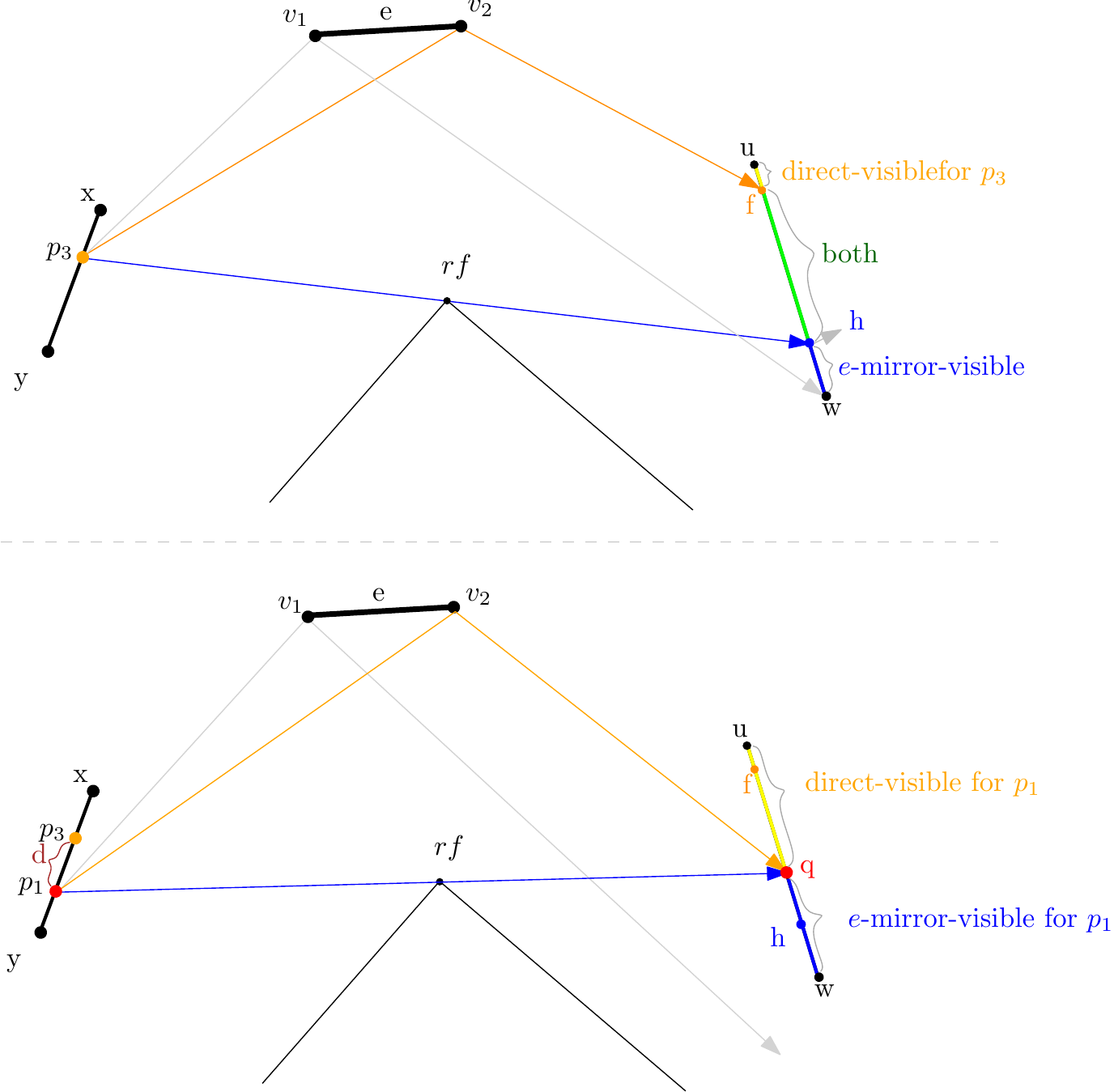}
\caption{For every mirror-edge, we can compute where direct visibility and mirror-visibility merge. To do this, we only have to solve a few algebraic equations.}
\label{fig:proof}
\end{center}
\end{figure}

As we consider any case  in which we have to merge direct visibility with a mirror-visibility and we have the corresponding reflex vertices and the endpoints of the mirror that have impact on the visibility, we can compute how to move on source to aim a specific position. Then, through algebraic equations, we can find $dmvm$ points. Note that the endpoints of mirror-edges (denoted by $v_{1}$ and $v_{2}$) and the corresponding reflex-vertices (denoted by $\rf$) are computed from the algorithm presented in \cite{tcs}.

So, all in all, we have a few algebraic computational steps that cost a constant time;  to compute all intervals on the source and on the target, as well as $v_{1}$ and $v_{2}$ vertices and the corresponding reflex-vertices, we just have to run the algorithm presented by \cite{tcs}. This algorithm has been proved to need linear time based on the complexity of $\P$. So, we are done. 

\end{proof}

\emph{Approach for case 3:}
\begin{enumerate}
\item
We compute every mirror-visible point on $target$ via every edge of $\P$. To do this, we use the algorithm mentioned in subsection 5.4 of the article \cite{tcs}; they use a segment denoted by $\seg{xy}$ as the viewer and a segment denoted by $\seg{uw}$ as the target. Based on the results of \cite{tcs}, every such mirror-visible part can be computed in the linear time. Furthermore, we will obtain the intervals on $ed$ that make those mirror-visible parts on the target. So, there will be at most $n$ intervals on $ed$ that can make at most $n$ mirror-visible parts on the target. Put the intervals on $ed$ in a set denoted by $\EdI$. Put
the mirror-visible parts on  $target$
in a set denoted by $\TI$. 
\item
Check the points in $\TI$ to find new intervals visible via direct-visibility and also, a mirror. These intervals determine the of $dmvm$ points projection. 
\item
Compute all $dmvm$ points.
 Add $dmvm$ points and the endpoints of the intervals in $\EdI$ to a set denoted by $\EdP$. 
\item
For every point in $\EdP$, compute the direct visible parts on $target$. We can add these direct visible parts to another set. Now, we can sweep on $ed$; consider the points on $\EdP$ as the events; the sweeper checks if in every event the direct visible parts in combination with the mirror-visible parts can make the whole $target$ visible.  For every consequent event, if both points could see the entire target and also they are connected by at least on interval in $\EdI$, then we add the part between these  two event points to a one that can see the whole target.
\end{enumerate}

The sweeper we use in case 3 needs $\mathcal{O}(n\log(n))$ time. Case 3 also runs for every edge of a target second-order-convex-region; so, the time complexity we need for this case is: $\mathcal{O}(n|target|\log(n))$; the complexity of the $|target|$ equals $|\SCR|$.

Note that from the previous cases, we already know the parts on $ed$ that make the target completely visible. So, if we reach such an interval, we can immediately attach it to the current interval we are computing.

At the end, considering all situations, we can find the intervals on $ed$ that can make $target$ completely visible.

%%%%%%%%%%%%%%%%%%%%%%%%%%%%%%%%%
\subsection{Decomposing a second-order-convex-region}
\label{sub.algo.intersect}
In this subsection, we see how to decompose a given second-order-convex-region into distinct and disjoint guarding-regions. This subsection is the complete version of Subsection~\ref{sub.algo.intersect-summary}. Here, we also prove Lemma~\ref{proof:decompose}.
From the previous steps of Algorithm~\ref{algo.main}, we know that a temp-sub-region is created based on the collision of two half-lines with the boundary of the given second-order-convex-region. The approach sweeps on the given second-order-convex-region. When the sweeper reaches a half-line that starts/ends a temp-sub-region, we add/eliminate the visible second-order-convex-region corresponding to that temp-sub-region to the visible-list of currently in progress guarding-regions. If the sweeper reaches an intersection point of many half-lines, there could be more than one second-order-convex-region that we should add or eliminate from the currently in progress guarding-regions. The approach is presented in details in the following:

\emph{Decomposition-Process:}
Consider a second-order-convex-region $scr_{i}$. To decompose $scr_{i}$, we use a sweep-line called $\cal{SL}$. The sweeping process starts from an edge $e$ of $scr_{i}$. The sweep-line $\cal{SL}$ first contains $e$ and moves on $scr_{i}$ parallel to $e$ till it reaches the end of $scr_{i}$, where $\cal{SL}$ sees all points of $scr_{i}$.

The final decomposition of $scr_{i}$ results in several newly created guarding-regions. The Decompose procedure puts the final guarding-regions in a set called $s_{i}$. 
Note that we have to compute a visible-list for every new guarding-region, $gr$. We already know that $\VL(gr \in scr_{i})$ contains $scr_{i}$. 
The sweep-line $\cal{SL}$ has a list of currently in progress guarding-regions that should decide where to end them. Denote this list by $\inPL$. Denote the number of guarding-regions in the current list of the sweep-line by $|\inPL|$. 
During the sweeping process, $\cal{SL}$ may encounter three cases: 1) It may reach a starting half-line $shl$,  2) it may reach an intersection-point of several $shl$ or $ehl$ half-lines,  or 3) it may reach an ending half-line $ehl$.

\begin{figure}[tp]%figure2
\begin{center}
\includegraphics[scale=0.6]{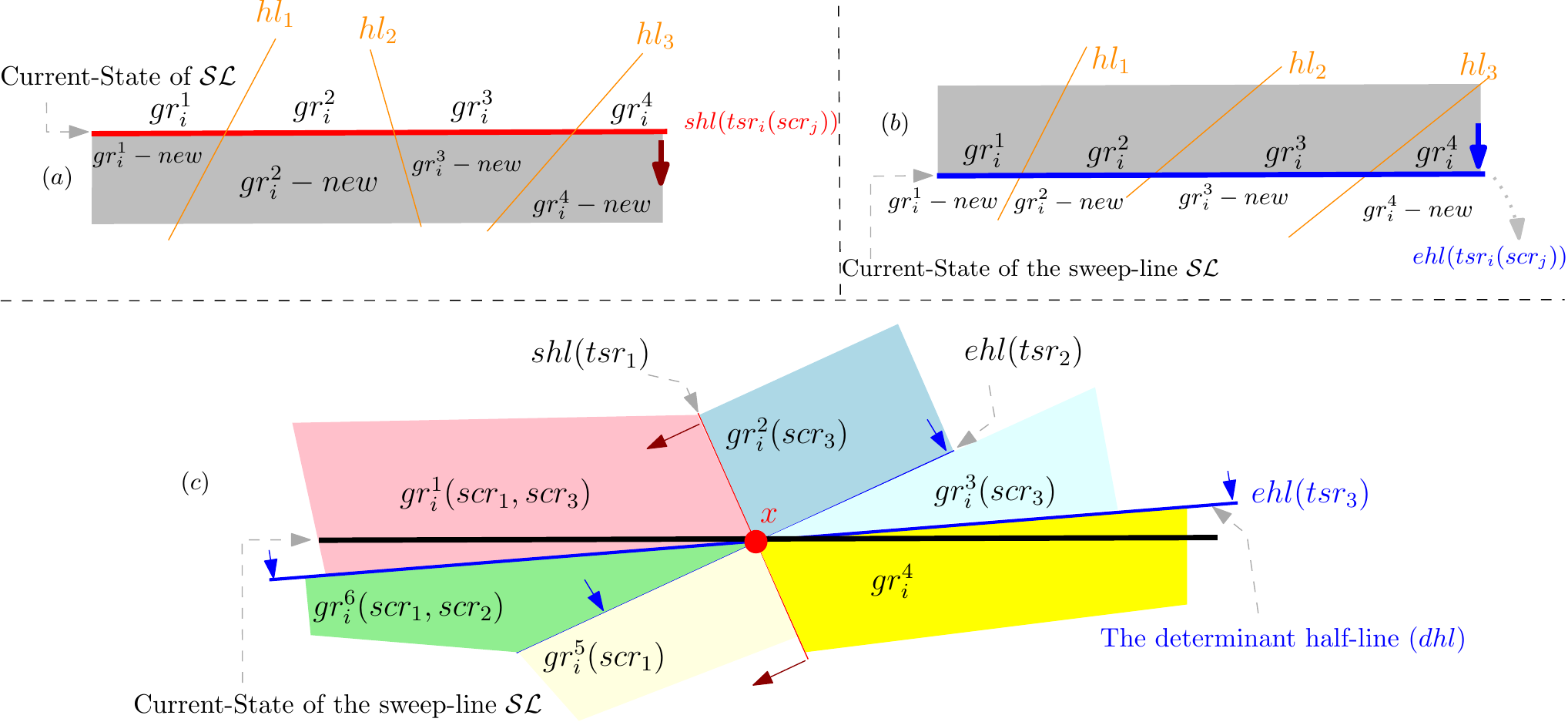}
\caption{
Figure $(a)$ shows that when $\cal{SL}$ reaches a $shl(tsr_{i}(scr_{j}))$ half-line, we should create new guarding-regions with new visible-lists. We have to add $scr_{j}$ to the visible-list of the guarding-regions that $\cal{SL}$ has just passed. Figure $(b)$ illustrates a similar situation where $\cal{SL}$ reaches an ending half-line $ehl(tsr_{i}(scr_{j}))$. In such a situation, we have to end the currently in process guarding-regions and create new ones with new corresponding visible-lists that do not contain $scr_ {j}$ in their visible-lists. In Figure, $(c)$ $\cal{SL}$may reach an intersection point, $x$, of several half-lines (either $shl$ or $ehl$). In this case, there is a determinative half-line $dhl$ that passes $x$. This half-line can be either a starting half-line or an ending one. Since the sweep-line $\cal{SL}$ has previously encountered $dhl$, the half-line $dhl$ does not make any new guarding-region itself. Nevertheless, the intersection-point $x$ makes new guarding-regions from the previously visited half-lines. As can be seen in Figure (c), $gr^6$ exhibited in green or $gr^5$ in white are the new guarding-regions starting from $x$. However, $gr^2$ and $gr^3$ are ended at $x$. For $gr^1$ and $gr^4$, it is sufficient to update their visible-lists. Based on the type of each half-line, we can easily create the new guarding-regions around $x$; also, we can determine their corresponding visible-lists. Note that if $dhl (tsr_{i}(scr_{j}))$ is parallel to the direction of $\cal{SL}$, then the sweep-line reaches a new half-line and a point; the creation process of the new guarding-regions is a combined procedure of the above-mentioned strategies. In such a case, we have to consider $scr_{j}$ as an additional or removal region in the visible lists of the new guarding-regions below $dhl(tsr_{i}(scr_{j}))$.}
\label{fig.decopposing.cases}
\end{center}
\end{figure}

Consider an arbitrary $scr_{k} \neq scr_{i}$ and suppose that $tsr \in scr_{i}$ can make $scr_{k}$ area-visible. 
There are three different cases. See Figure~\ref{fig.decopposing.cases} as an example.
In this figure, the numbers on power of $gr$ are normal indicators to itemize different guarding-regions. This figure illustrates three different cases a sweep-line may encounter while tracing a second-order-convex-region $(scr_{i})$ to decompose it into guarding-regions. The three cases:

 1,2) If $\cal{SL}$ reaches a starting half-line $shl(tsr_{i}(scr_{k}))$ or an ending half-line $ehl(tsr_{i}(scr_{k}))$;
 
   Do these steps:
    \begin{enumerate}
        \item  End all guarding-regions in $\inPL$.
        That is because all previously submitted guarding-regions in $\inPL$ are ended by the new line. 
        \item  Add every newly ended guarding-region $gr$ to $s_{i}$ (the output set).
        \item Based on the number of crossing lines with the half-line $shl(tsr_{i})/ehl(tsr_{i})$, create new guarding-regions in $\inPL$. Each of such new guarding-regions starts from $e/shl(tsr)$ and has its own boundary.
        \item Create a visible-list for each new $gr \in \inPL$ ($\VL(gr)$) with the initiation value of $scr_{i}$.
       \item For each $gr \in \inPL$, if $gr_{x}$ has $gr_{y}$; on the other side; (the already swept side) of $e/shl(tsr)$, set $\VL(gr_{x}) \longleftarrow \VL(gr_{y})$.
       Note that we only check one line at a moment, so only one temp-sub-region and consequently, the visibility of one second-order-convex-region are checked, and other visible second-order-convex-regions are still visible.

       \item For each new $gr \in \inPL$,
       
       {\bf if} $\cal SL$ reaches $shl(tsr_{i}(scr_{k}))$;

       $\ \ \ $Add $scr_{k}$ to $\VL(gr)$.
       
       {\bf else} $\cal SL$ reaches $ehl(tsr_{i}(scr_{k}))$;
       
        $\ \ \ $Remove $scr_{k}$ from $\VL(gr)$.
       
    \end{enumerate}
  3)
   If $\cal{SL}$ reaches an intersection-point $x$ of $k$ $shl$ half-lines and $k'$ $ehl$ half-lines\\, without loss of generality, suppose all of these half-lines are on distinct lines, and assume that the sweep-line is slightly tilted, so that the intersections are considered one by one. 
   We already know that $|inPL|=\frac{k+k'}{2}$ (see \autoref{fig.decopposing.cases}(c)). So, all half-lines are already seen by $\cal{SL}$. The most recently visited half-line is either a $shl$ or a $ehl$. Again, w. l. o. g. suppose it is $shl(tsr_{i}(scr_{k}))$, and $shl(tsr)$ is parallel to the direction of the sweep-line $\cal{SL}$.
   
   Follow the following steps:

        \begin{enumerate}
        \item  End every guarding-region in $\inPL$. These guarding-regions are above $shl(tsr)$, where $\cal{SL}$ is moving before it reaches $x$.
        \item  Add every newly ended guarding-region to $s_{i}$.
        \item  Based on the lines crossing $x$, create new guarding-regions in $\inPL$.
        \item Create a visible-list for each new guarding-region ($gr$) in $\inPL$ ($\VL(gr)$).
        \item Add $scr_{k}$ to every $\VL(gr)$. %visible-lists.
        \item There are nv  newly guarding-regions created in $s_{i}$. For each $gr$, check it with each half-line $hl$ intersected in $x$.
        
        (a) If $hl$ is a starting half-line $shl(tsr_{i}(scr_{k}))$ and $gr$ lies between $shl(tsr_{i}(scr_{k}))$ and $ehl(tsr_{i}(scr_{k}))$, then add $scr_{k}$ to $\VL(gr)$.
        
        (b) If $hl$ is an ending half-line $ehl(tsr_{i}(scr_{k}))$ and $gr$ is not between $shl(tsr_{i}(scr_{k}))$ and $ehl(tsr_{i}(scr_{k}))$ anymore, then remove $scr_{k}$ from $\VL(gr)$. 
        
        \end{enumerate}

Set $s_{i}$ from the above-mentioned approach as the output, which is the decomposed set of a given second-order-convex-region.

%%%%%%%%%%%%%%%%%%%%%%%%%%%%%%%%%%%%%%%%%%%%%%%%%%%%%%%%%%%%%%%%%%%%%%%%%%%%%%%%%%%%%%%%%%%%%%%%%%%%%%%%%%%%%%%%%%%%%%%%%%%%%%%%%%%%%%%%%%%%%%%%%%%%%%%%%%%%%%%%%%%%%%%%%%%%%%%%%%%%%%%%%%%%%%%%%%%%%%%%%%%%%%%%%%%%%

\subsection{Proof of Lemma~\ref{lem:2-app}}
\label{sub; proof}
%\begin{lemma}
$|\OPT(\S)|$ $\leq \frac{3}{2}$ $|\OPT^{*}|$
%\end{lemma}
\begin{proof}
Every guard $g_{opt}$ in $\OPT^{*}$ covers a sub-area of $\P$ denoted by $\VP(g_{opt})$, which is in fact, the visibility polygon of $g_{opt}$.
Pick an arbitrary point $p$ in a guarding-region $gr$. Set $g_{s}$ to be in the position specified by $p$. The guard $g_{s}$ can see a sub-area of $\P$ denoted by $\VP(g_{s})$.
The proposed algorithm tries to estimate $\VP(g_{s})$ with a sub-set of second-order-convex-regions. This sub-set is determined by the visible-list of a guarding-region ($\VL(gr)$). 
Choosing a guarding-region in $\S$ corresponds to selecting a position for a guard.

The guards in $\OPT^{*}$ also cover all second-order-convex-regions.
Remember that every guard Algorithm~\ref{algo.main} counts only on the complete visibility of second-order-convex-regions.
If some guards in $\OPT^{*}$ can see more second-order-convex-regions completely, then we can choose the same guards to be in $\OPT(\S)$ to cover the same second-order-convex-regions. Otherwise, $|\OPT^{*}|$ and $|\OPT(\S)|$ would have the same number of guards. 
So, we only have to focus on situations where at least one second-order-convex-region is covered by more than one previously chosen guard in $\OPT^{*}$, so that this second-order-convex-region is covered by the union of the partial visibility of those guards in $\OPT^{*}$. In $\OPT(\S)$, we have to count on extra guards to cover partially visible second-order-convex-regions completely.
We will see that since reflex vertices are connected to the intersection points from the criss-cross of the lines of the polygon,  those situations where a few guards in $\OPT^{*}$ can see many second-order-convex-regions partially cannot happen. However, depending on the positions reflex vertices and guards may choose, the following is needed to be investigated for the proof to be complete.

The problem happens when two or more guards in $\OPT*$ make some second-order-convex-regions completely visible by using the partial visibility of each guard. In $\OPT(\S)$, to cover those second-order-convex-regions, there must be some extra guards, so that each of the second-order-convex-regions gets area-visible by a specific guard (see Figure~\ref{fig.analysis}(a)).

Consider a guard $g_{opt}$ in $\OPT^{*}$; for $g_{opt}$, to see a second-order-convex-region partially, $g_{opt}$'s visibility must be blocked by a reflex-vertex.

Suppose we have $k$ guards in $\OPT^{*}$ that can see $k’ \geq  2$ second-order-convex-regions together and via the partial visibility of each one. See Figure~\ref{fig.analysis}(b). Consider a guard $g^{i}_{opt}$ that cannot see a second-order-convex-region $scr^{'}$ completely because of $\rf_{1}$'s blocking. The rest of $scr^{'}$ should get visible by some other guards $g^{k}_{opt} \ k\neq i$ in $\OPT^{*}$. 
Denote the $g^{i}_{opt}$-visible part of $scr^{'}$ by $\Vp(g^{i}_{opt},scr^{'})$. 
Consider another guard $g^{j}_{opt}$ that is responsible for the visible part very close to $\rf_{1}$.
Since $\rf_{1}$ is a reflex vertex, from the first steps of Algorithm~\ref{algo.main}, we know that there must be a vertex or an intersection point, denoted by $v_1$, which is connected to $\rf_{1}$.
Choose $v_1$, so that the line containing $\seg{v_{1}\rf_{1}}$ contains an edge of $\Vp(g^{j}_{opt},scr^{'})$. 
Note that w.l. o. g we suppose that the visible part $\Vp(g^{j}_{opt},scr^{'})$ is only completely visible to $g^{j}_{opt}$ and no other guards in $g^{k}_{opt} \ k\neq j$.  Again, without the loss of generality, suppose $g^{i}_{opt}$ and $g^{j}_{opt}$ can see another second-order-convex-region $scr^{''}$ partially because of blocking other reflex-vertices. Choose $\rf_{3}$ to block the visibility of $g^{i}_{opt}$ not to see $scr^{''}$ completely. Consider the point $v_4$, so that the unique line that contains $\seg{v_{4}\rf_{4}}$ crosses an edge of $\Vp(g^{i}_{opt},scr^{''})$. 
The guard $g^{i}_{opt}$ is responsible to see both $\Vp(g^{i},scr^{''})$ and $\Vp(g^{i},scr^{'})$ areas and perhaps, more such areas. 
Also, the guard $g^{j}_{opt}$ is responsible for $\Vp(g^{j},scr^{''})$ and $\Vp(g^{j},scr^{'})$ areas.

As mentioned above, we know that $g^{i}_{opt}$ should see both $\Vp(g^{i},scr^{'})$ and $\Vp(g^{i},scr^{''})$ areas or more. However, consider the unique line that contains $\seg{v_{4}\rf_{4}}$, and the unique line that contains $\seg{v_{1}\rf_{1}}$. These lines must intersect in a point. That is because there is a point with the name of $g^{i}_{opt}$ in the optimal solution that sees those areas. The lines cannot be parallel because there would be no point to see both $\VP(g^{i}_{opt},scr^{''})$ and $\VP(g^{i}_{opt},scr^{'})$ areas simultaneously. So, there is an intersection point $p$. Algorithm~\ref{algo.main} connects $p$ to every reflex vertex of $\P$. Note that the algorithm does this whether $p$ is inside $\P$ or not. So, for every $g_{opt}$ corresponding to every partially visible area $\Vp(g_{opt},scr)$, there is at least one second-order-convex-region. This contradicts a $\Vp(g_{opt},scr)$ of a guard to be a partial visible part of a second-order-convex-region.

 In the mentioned proof, we picked two arbitrary guards from $\OPT^{*}$, proving that they cannot see two or more second-order-convex-regions via simultaneous partial visibility. The fact that reflex-vertices were connected to the intersection points of lines during the construction of second-order-convex-regions helped us in getting the proof. However, every two guards may see {\bf one} second-order-convex-region together.
 For every such two guards in $\OPT^{*}$, we may need an additional guard in $\OPT(\S)$ to make the corresponding second-order-convex-region area-visible. This extra guard is denoted by $g_{e}$ in Figure~\ref{fig.analysis}(a).
 So, in the worst-case corresponding to every two guards in $\OPT^{*}$, we might need three guards in $\OPT(\S)$.  Thus, we are done with the proof.

\begin{figure}[tbp]%figure2
\begin{center}
\includegraphics[scale=0.7]{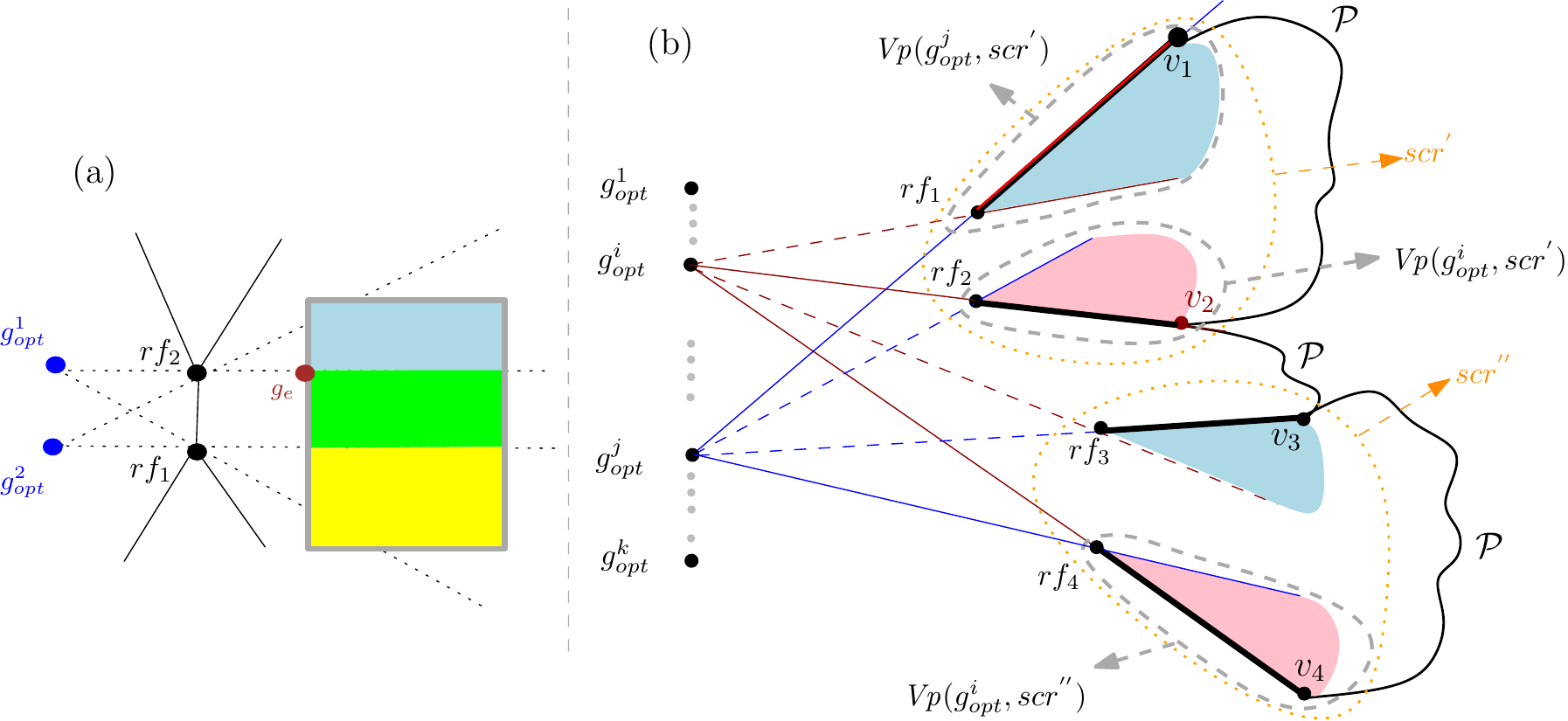}
\caption{$(a)$; in this case, we can add an extra guard $g_e$ and the whole second-order-convex-region is covered in $\OPT(\S)$.
 $(b)$ This case reveals that if there are more than one second-order-convex-region all covered by the integration of the partial visibility of a few guards in $\OPT^{*}$, in $\OPT(\S)$, we have the same number of guards in  place of each of those guards in $\OPT^{*}$. In other words, corresponding to the areas covered by those optimal guards, in such a situation, there are second-order-convex-regions that can get entirely visible by guards in the same positions. Furthermore, those second-order-convex-regions cover the whole surface of those areas. So, these guards are covered in $\OPT(\S)$ too.
 Figure $(c)$; blue and pink second-order-convex-regions indicate the second-order-convex-regions in the right side areas. Two areas are covered with two guards in $\OPT^{*}$. Thus, we have four second-order-convex-regions that make these areas to get covered by the same number of guards in $\OPT(\S)$.  }
\label{fig.analysis}
\end{center}
\end{figure}
\qed
\end{proof}

In fact, we proved that the way Algorithm~\ref{algo.main} decomposed $\P$ into second-order-convex-regions and guarding-regions could make it so close to the real optimal solution.
\section{List of Notations}
To better present and convince of the reader, we list the frequent notations used throughout the paper.

$scr$ = a second-order-convex-region

$gr$ = a guarding-region

$\rf$ = a reflex-vertex

$\VP$ = visibility polygon of a point or a guard

$\Vp$ = visible part of a guard inside a second-order-convex-region

$\VL$ = visible-list of a guarding-region

$\S$ = the set of all guarding-regions after the decomposition of $\P$ into guarding-regions by beginning Step 4 of Algorithm~\ref{algo.main}

$\OPT^{*}$ = the set of guards in the optimal solution of the point guarding problem

$\OPT(\S)$ = the set of guards in the optimal solution that can be obtained from $\S$.
\end{document}